\documentclass[%
 reprint,
 amsmath,amssymb,
 aps,
]{revtex4-1}

\usepackage{graphicx}
\usepackage{dcolumn}
\usepackage{bm}%
\usepackage{hyperref}
\usepackage{tikz}
\usepackage{amsthm}
\usepackage{float}
\usepackage{comment}
\usepackage{algorithm}
\usepackage{algpseudocode}
\usepackage{caption}
\usepackage{braket}

\usepackage[pagewise, displaymath, mathlines]{lineno}

\usepackage{multirow}

\usepackage[usenames,dvipsnames]{pstricks}
\usepackage{epsfig}

\newtheorem{theorem}{Theorem}[section]
\newtheorem*{theorem*}{Theorem}
\newtheorem{proposition}[theorem]{Proposition}
\newtheorem*{proposition*}{Proposition}
\newtheorem{corollary}[theorem]{Corollary}
\newtheorem*{corollary*}{Corollary}
\renewcommand\thetheorem{\arabic{section}.\arabic{theorem}} 

\theoremstyle{definition}
\newtheorem{definition}[theorem]{Definition}

\newtheorem*{lemma*}{Lemma}

\newtheorem*{problem*}{Problem}

\DeclareMathOperator{\Tr}{Tr}
\newcommand{\com}[1]{\textcolor{red}{\textbf{[[} #1 \textbf{]]}}}

\begin{document}


\title{Simulation of Qubit Quantum Circuits via Pauli Propagation}

\author{Patrick Rall}
\author{Daniel Liang}
\author{Jeremy Cook}
\author{William Kretschmer}%
\affiliation{Quantum Information Center, University of Texas at Austin}

\date{\today}

\begin{abstract}
We present novel algorithms to estimate outcomes for qubit quantum circuits. Notably, these methods can simulate a Clifford circuit in linear time without ever writing down stabilizer states explicitly. These algorithms outperform previous noisy near-Clifford techniques for most circuits. We identify a large class of input states that can be efficiently simulated despite not being stabilizer states. The algorithms leverage probability distributions constructed from Bloch vectors, paralleling previously known algorithms that use the discrete Wigner function for qutrits.

\end{abstract}

\maketitle


\section{Introduction}

Simulating quantum circuits on classical hardware requires large computational resources. Near-Clifford simulation techniques extend the Gottesmann-Knill theorem to arbitrary quantum circuits while maintaining polynomial time simulation of stabilizer circuits. Their runtime analysis gives rise to measures of non-Cliffordness, such as the robustness of magic \cite{resource}, magic capacity \cite{seddon}, sum-negativity \cite{vmge13}. These algorithms evaluate circuits by estimating the mean of some probability distribution via the average of many samples, a process with favorable memory requirements and high parallelizability.

Previous work \cite{bennink, resource} gives an algorithm based on quasiprobability distributions over stabilizer states; we refer to this algorithm as `stabilizer propagation'. In contrast to techniques based on stabilizer rank \cite{gosset, extent}, stabilizer propagation is appealing for simulation of NISQ-era hardware \cite{nisq} because it can simulate noisy channels. Moreover, depolarizing noise decreases the number of samples required, measured by robustness of magic and the magic capacity. However, bounding the number of required samples can be expensive: For example, the magic capacity of a 3-qubit channel is defined as a convex optimization problem over 315,057,600 variables \cite{enums, seddon}.

Pashayan et al.\ \cite{pash} showed that in qutrit systems, the discrete Wigner function provides a simpler simulation strategy. This strategy takes linear time to sample, and the number of samples required (measured by the sum-negativity) is tractable to compute for small systems. However, discrete Wigner functions do not yield efficient simulation of qubit Clifford circuits \cite{rbdobv15}.

Our main result is that Bloch vectors yield simulation strategies for qubit circuits, similar to those in Pashayan et al. We present two algorithms, which we individually call \textbf{Schr\"odinger propagation} and \textbf{Heisenberg propagation}, and collectively call \textbf{Pauli propagation techniques}. They have several surprising properties:
\begin{enumerate}
    \item They yield linear time simulation for qubit Clifford circuits without writing down stabilizer states.\vspace{1mm}
    \item Schr\"odinger propagation can efficiently simulate a new family of quantum states called `hyper-octahedral states' which is significantly larger than the set of stabilizer mixtures in terms of the Hilbert-Schmidt measure. \item The runtime of Heisenberg propagation does not depend on the input state at all.
    \item Non-Cliffordness in both algorithms is measured via the stabilizer norm, which is a lower bound to the robustness of magic. This gives Pauli propagation techniques a strictly lower runtime than stabilizer propagation for all input states and most channels.
\end{enumerate}

\onecolumngrid
\vfill


\begin{table}[b]
\setlength{\tabcolsep}{6pt}
\renewcommand{\arraystretch}{1.1}

    \textbf{Table: Circuit components that can be simulated efficiently}
    \vspace{2mm}\\

\begin{tabular}{r||l | l | l}
 & \multicolumn{1}{c}{Previous work \cite{bennink, resource}} & \multicolumn{2}{c}{This work: Pauli propagation algorithms} \\
 & \textbf{Stabilizer propagation}  & \textbf{Heisenberg propagation}  & \textbf{Schr\"odinger propagation} \\ \hline \hline
    What input states &   \multirow{2}{*}{Stabilizer mixtures} & Any separable state,  & Hyper-octahedral states,  \\
are efficient to simulate? & &  Stabilizer mixtures & Noisy states reduce runtime  \\ \hline
    Depolarized $T$ gate  & Efficient when fidelity $ \lessapprox 0.551$  &  \multicolumn{2}{c}{Efficient when fidelity $\leq 2^{-1/2} \approx 0.707$} \\ \hline
    Reset channels   & Pauli reset channels efficient   & All reset channels efficient    & \multirow{2}{*}{Generally inefficient} \\ \cline{1-3}
    Adaptive gates    & Adaptive Cliffords efficient & Generally inefficient  & \\ \cline{1-4}
    Marginal observables &  \multirow{2}{*}{Efficient}  & \multirow{2}{*}{Efficient}  & \multirow{2}{*}{Generally inefficient}  \\\cline{1-1} 
Pauli observables     &        & & \\ 
\end{tabular}
    \caption*{Summary of the results of Section III. All algorithms take polynomial time to sample, but the number of samples scales exponentially in the number of \emph{inefficient} circuit components. \emph{Efficient} components do not increase runtime.}
\end{table}
\twocolumngrid

\hspace{5mm}\\

We describe these algorithms in Section II. In Section III we perform a detailed comparison of Schr\"odinger, Heisenberg and stabilizer propagation which we summarize in the table below. In Section IV we briefly discuss the implications of the algorithms for resource theories of Cliffordness. This work is intended to supersede \href{https://arxiv.org/abs/1804.05404}{quant-ph/1804.05404}.

\section{\label{sec:intro}Algorithms}
In this section we describe two algorithms for estimating the expectation value of observables at the end of a quantum circuit. Schr\"odinger propagation involves propagating states forward though the circuit and taking inner products with the final observables. Heisenberg propagation involves propagating observables backward though the circuit and taking inner products with the initial states. At every step, both procedures sample from an unbiased estimator for the propagated state/observable that is distribution over Pauli matrices. 

\subsection{Sampling Pauli Matrices}
The workhorse of both protocols is a subroutine that samples a random scaled tensor product of Pauli matrices as a proxy for an arbitrary $n$-qubit Hermitian matrix $A$. Let $\mathcal{P}_n = \{\sigma_1 \otimes \cdots \otimes \sigma_n : \sigma_i \in \{I, \sigma_X, \sigma_Y, \sigma_Z\}\}$ denote the set of $n$-qubit Pauli matrices. We define a pair of completely dependent random variables $\hat \sigma \in \mathcal{P}_n$ and $\hat c \in \mathbb{R}$ that satisfy $\mathbb{E}\left[\hat c \cdot \hat \sigma\right] = A$:
\begin{align}
\label{eq:sigma_hat}\hat \sigma(A) &= \sigma \text{ with prob. } \frac{\left|\text{Tr}(\sigma A)\right|}{ 2^n \cdot \mathcal{D}(A)} \text{ for each }\sigma \in \mathcal{P}_n,\\
\label{eq:c_hat}\hat c(A) &= \mathrm{sign}\left( \mathrm{Tr}(\hat \sigma(A) A)\right) \cdot \mathcal{D}(A).
\end{align}

The quantity $\mathcal{D}(A)$ is a normalization constant that makes $\frac{\left|\text{Tr}(\sigma A)\right| }{ 2^n \cdot \mathcal{D}(A)}$  for $\sigma \in \mathcal{P}_n$ a probability distribution.

\begin{definition}
The \textbf{stabilizer norm} $\mathcal{D}(A)$ is:
\begin{equation}\label{eq:stabnorm}\mathcal{D}(A) = \frac{1}{2^n} \sum_{\sigma \in \mathcal{P}_n} \left|\mathrm{Tr}(\sigma A) \right|.
\end{equation}
\end{definition}

The product of the random variables $\hat c(A) \cdot \hat \sigma(A)$ is an unbiased estimator for $A$ because the Pauli matrices form an operator basis for Hermitian matrices:
\begin{align}
\mathbb{E}[\hat c(A) \cdot \hat \sigma(A)] &= \sum_{\sigma \in \mathcal{P}_n} \frac{\left|\text{Tr}(\sigma A)\right|}{ 2^n \cdot \mathcal{D}(A)}  \cdot \text{sign}\left( \text{Tr}(\sigma A) \right) \cdot \mathcal{D}(A)\cdot \sigma\nonumber\\
&= \sum_{\sigma \in \mathcal{P}_n} \frac{\text{Tr}(\sigma A)}{2^n} \cdot \sigma = A.
\end{align}

The time to compute the probabilities and sample from the distributions scales exponentially with the number of qubits of $A$. We say $A$ has \textbf{tensor product structure} if it can be written as a tensor product of several operators, each of which acts on a constant number of qubits:
$$A = A_1\otimes A_2 \otimes \cdots$$

\noindent Then one can observe that:
$$\hat\sigma(A) = \hat\sigma(A_1)\otimes \hat\sigma(A_2)\cdots \text{ and } \hat c(A) = \hat c(A_1) \cdot \hat c(A_2)\cdots$$
Since each $A_i$ acts on a constant number of qubits, each of the probability distributions for $\hat \sigma(A_i), \hat c(A_i)$ can be computed and sampled from in constant time. So $\hat\sigma(A)$ and $\hat c(A)$ can be sampled from in linear time if $A$ has tensor product structure, even if $A$ acts on many qubits.

\subsection{Schr\"odinger Propagation}

Suppose we want to apply a sequence of channels $\Lambda_1,\ldots,\Lambda_k$ to an $n$-qubit state $\rho_0$. These operations are given as a quantum circuit, so $\rho_0$ has tensor product structure and each of the $\Lambda_i$ non-trivially act on a constant-size subset of the qubits. Let $\rho_i$ be the state after applying the first $i$ channels:
\begin{equation}\rho_i = \Lambda_i(\Lambda_{i-1}(\cdots\Lambda_1(\rho_0)))\end{equation} 

We are given an observable $E$ which also has tensor product structure. We want to estimate the expectation of $E$ on the final state:
\begin{equation}\langle E\rangle  = \text{Tr}\left( E \rho_k \right)  = \text{Tr}\left( E \Lambda_k(\Lambda_{k-1}(\cdots\Lambda_1(\rho_0))) \right)\label{eq:schrgoal}\end{equation}

We apply the sampling procedure defined by \eqref{eq:sigma_hat} and \eqref{eq:c_hat}  to $\rho_0$. We define $\hat \sigma(\rho_0) = \hat\sigma_0$ and $\hat c(\rho_0) = \hat c_0$. Their product $\hat c_0 \cdot \hat \sigma_0$ is an unbiased estimator for $\rho_0$.

Given an unbiased estimator $\hat c_i \cdot \hat \sigma_i$ for $\rho_i$, we will obtain an unbiased estimator $\hat c_{i+1} \cdot \hat \sigma_{i+1}$ for $\rho_{i+1}$. Apply $\Lambda_{i+1}$ to $\hat c_i \cdot \hat \sigma_i$ and use linearity of $\Lambda_{i+1}$:
$$  \mathbb{E} \left[ \Lambda_{i+1}(\hat c_i \cdot \hat \sigma_i)  \right] =\Lambda_{i+1}(\mathbb{E} \left[ \hat c_i \cdot \hat \sigma_i \right])  = \rho_{i+1}$$

We have $\Lambda_{i+1}(\hat c_i \cdot \hat \sigma_i)  = \hat c_i \cdot   \Lambda_{i+1}( \hat \sigma_i) $. Since $\Lambda_{i+1}$ acts non-trivially on a constant-size subset of the qubits, $\Lambda_{i+1}( \hat \sigma_i)$ has tensor product structure and we can sample using \eqref{eq:sigma_hat} and \eqref{eq:c_hat} again. Let:
\begin{equation}\hat \sigma_{i+1} = \hat \sigma\left(\Lambda_{i+1}( \hat \sigma_i)\right) \text{ and } \hat c_{i+1} = \hat c_{i}\cdot \hat c\left(\Lambda_{i+1}( \hat \sigma_i)\right)  \end{equation}

Now we have $\hat c_{i+1} \cdot \hat \sigma_{i+1}$, an estimator for $\rho_{i+1}$, and can recursively obtain $\hat c_{k} \cdot \hat \sigma_{k}$ for $\rho_k$. Since $E$ and $\hat\sigma_k$ have tensor product structure, we can efficiently obtain their trace inner product. The protocol yields a sample from the distribution in time linear in $k+n$:
\begin{equation}\text{Output: sample from } \hat c_{k} \cdot \text{Tr}( \hat \sigma_k E ) \label{eq:shrout}\end{equation}
This distribution estimates the target quantity:
$$\mathbb{E}\left[\hat c_{k} \cdot \text{Tr}( \hat \sigma_k E )  \right] = \text{Tr}( \mathbb{E}\left[\hat c_{k} \cdot \hat \sigma_k\right] E ) = \text{Tr}\left(\rho_k E  \right) = \langle E\rangle  $$

We estimate the mean of $ \hat c_{k} \cdot \text{Tr}( \hat \sigma_k E ) $ by taking the average of $N$ samples. The Hoeffding inequality \cite{hoeffding} provides a sufficient condition on $N$ for an additive error $\varepsilon$ with probability $1-\delta$ in terms of the range of the distribution:
\begin{equation}N \geq \frac{1}{2\varepsilon^2} \cdot \ln \frac{2}{\delta} \cdot (\text{range})^2\label{eq:hoeff}\end{equation}
The range of the output distribution is bounded by twice the maximum magnitude of the output distribution \eqref{eq:shrout}.
\begin{equation} \text{range} \leq 2 \cdot \left|\hat c_k \cdot \mathrm{Tr}(\hat \sigma_k E) \right| \leq 2\cdot  |\hat c_k| \cdot \max_{\sigma \in \mathcal{P}_n}\left|\mathrm{Tr}(\sigma E)\right|\end{equation}

\noindent Observe that $\hat c(A) = \pm \mathcal{D}(A)$, so:
\begin{align}
    |\hat c_{i+1}| &= |\hat c_{i}|\cdot| \hat c\left(\Lambda_{i+1}( \hat \sigma_i)\right)| \nonumber\\
    &= |\hat c_{i}|\cdot \mathcal{D}(\Lambda_{i+1}(  \hat \sigma_i)) \nonumber \\
    &\le |\hat c_{i}| \cdot \max_{\sigma \in \mathcal{P}_n}\mathcal{D}(\Lambda_i(\sigma))\label{eq:ccost}
\end{align}
Intuitively, $\mathcal{D}$ measures the ``cost'' of a Hermitian matrix in this algorithm. The above motivates a corresponding notion of the ``cost'' of a channel:

\begin{definition}
The \textbf{channel stabilizer norm} $\mathcal{D}(\Lambda)$ is defined by:
\begin{equation} \mathcal{D}(\Lambda) = \max_{\sigma \in \mathcal{P}_n} \mathcal{D}(\Lambda(\sigma))\end{equation}
\end{definition}

\noindent Expanding the recursion in (\ref{eq:ccost}) we obtain the bound:
\begin{equation}\left|\hat c_k \cdot \mathrm{Tr}(\hat \sigma_k E) \right|  \le \underbrace{\vphantom{\prod_{i=1}^k}\mathcal{D}(\rho_0)}_{(1)} \cdot \underbrace{\prod_{i=1}^k \mathcal{D}(\Lambda_i)}_{(2)} \cdot \underbrace{\vphantom{\prod_{i=1}^k}\left|\max_{\sigma \in \mathcal{P}_n}\mathrm{Tr}(\sigma E)\right|}_{(3)}\end{equation}
The number of samples $N$ scales with the square of the above quantity. Thus, the cost of Schr\"odinger propagation on a circuit breaks into three parts: (1) the cost of the initial state, (2) the cost of each channel, and (3) the cost of the final observable.

Here are two observations:
\begin{itemize}
    \item Say $\rho_0 = \rho^{\otimes m}$, so $\mathcal{D}(\rho_0) = \mathcal{D}(\rho)^m$. For many $\rho$ with short Bloch vectors, the cost $\mathcal{D}(\rho)$ can be strictly less than 1, meaning more copies of $\rho$ result in an exponential runtime improvement from cost term (1).
    \item Often we are interested in observables $E_\text{local}$ that act only on a small subset of the output qubits. Then $E$ is a tensor product of linearly many identity matrices and $E_\text{local}$, resulting in an exponential runtime blowup from cost term (3).
\end{itemize}

Loosely speaking, Schr\"odinger propagation works well when the input qubits are noisy and all output qubits are measured, like some supremacy circuits \cite{suprem}.

\subsection{Heisenberg Propagation}

Heisenberg propagation involves propagating the observable $E$ backwards through the circuit and taking the inner product with the initial state $\rho_0$. To do so we utilize the \textbf{channel adjoint} $\Lambda^\dagger$ which satisfies:  
\begin{equation}\text{Tr}(E \Lambda(\rho)) = \text{Tr}(\Lambda^\dagger(E)\rho)\end{equation}


\noindent Applying this to \eqref{eq:schrgoal}, our goal is to estimate:
\begin{align}\langle E\rangle  = \text{Tr}\left( \rho_0 \Lambda^{\dagger}_1(\cdots\Lambda^{\dagger}_{k-1}(\Lambda^{\dagger}_k(E))) \right) =  \text{Tr}\left( \rho_0 E_1 \right) \nonumber\\
    \text{where }E_i = \Lambda^{\dagger}_{i}(\Lambda^{\dagger}_{i+1}(\cdots\Lambda^{\dagger}_{k-1}(\Lambda^{\dagger}_k(E)))) 
\end{align}

For Heisenberg propagation we will define $\hat c_i, \hat \sigma_i$ differently from Schr\"odinger propagation. We use the sampling procedure defined by \eqref{eq:sigma_hat} and \eqref{eq:c_hat} and obtain $\hat \sigma(E) = \hat \sigma_{k+1}$ and $\hat c(E) = \hat c_{k+1}$. Then $\hat c_{k+1} \cdot \hat \sigma_{k+1}$ is an unbiased estimator for $E$.

With an unbiased estimator $\hat c_{i+1} \cdot \hat \sigma_{i+1}$ for $E_{i+1}$ we can obtain an unbiased estimator  $\hat c_{i} \cdot \hat \sigma_{i}$ for $E_{i}$ from $\Lambda^{\dagger}_i(\hat c_{i+1} \cdot \hat \sigma_{i+1}) = \hat c_{i+1} \cdot \Lambda^{\dagger}_i( \hat \sigma_{i+1})$. Since $\Lambda^{\dagger}_i( \hat \sigma_{i+1})$ has tensor product structure we can sample using \eqref{eq:sigma_hat} and \eqref{eq:c_hat}, and obtain:
\begin{equation}\hat \sigma_{i} = \hat \sigma(\Lambda_i^\dagger(\hat\sigma_{i+1}))\text{ and } \hat c_{i} = \hat c_{i+1} \cdot \hat c(\Lambda_i^\dagger(\hat\sigma_{i+1}))\end{equation}

This operation is iterated until we obtain $\hat c_{1}\cdot \hat \sigma_1$, an unbiased estimator for $E_1$. Since $\rho_0$ has tensor product structure we can compute the trace inner product and produce a sample, again in time linear in $k+n$:
\begin{equation}\text{Output: sample from } \hat c_{1} \cdot \text{Tr}( \hat \sigma_1 \rho_0 ) \label{eq:heisout}\end{equation}
This estimates the target quantity:
$$\mathbb{E}\left[\hat c_{1} \cdot \text{Tr}( \hat \sigma_1 \rho_0 )  \right] = \text{Tr}( \mathbb{E}\left[\hat c_{1} \cdot \hat \sigma_1\right] \rho_0 ) = \text{Tr}\left(E_1 \rho_0  \right) = \langle E\rangle  $$

To bound the number of samples $N$ we bound the maximum magnitude of \eqref{eq:heisout} and utilize Hoeffding's inequality \eqref{eq:hoeff}. Since $\rho_0$ is a quantum state, we always have $\max_{\sigma \in \mathcal{P}_n}\left|\mathrm{Tr}(\sigma \rho_0)\right| = 1$ since the eigenvalues of $\sigma$ are $\pm 1$. This leaves the recursion relation:

\begin{align}
  \left|\hat c_1 \cdot \mathrm{Tr}(\hat \sigma_1 \rho_0) \right| \leq |\hat c_1| &= |\hat c_{i+1}|\cdot \left| \hat c\left(\Lambda^\dagger_{i}( \hat \sigma_{i+1})\right)\right| \nonumber\\
    &= |\hat c_{i+1}|\cdot \mathcal{D}(\Lambda^\dagger_{i}( \hat \sigma_{i+1})) \nonumber \\
    &\le |\hat c_{i+1}| \cdot \max_{\sigma \in \mathcal{P}_n}\mathcal{D}(\Lambda^\dagger_i(\sigma))\nonumber\\
    &= |\hat c_{i+1}| \cdot \mathcal{D}(\Lambda_i^\dagger) \label{eq:hccost}
\end{align}

\noindent Expanding the recursion we obtain the bound:
\begin{equation}\left|\hat c_1 \cdot \mathrm{Tr}(\hat \sigma_1 \rho_0) \right|  \le \underbrace{\vphantom{\prod_{i=1}^k}\mathcal{D}(E)}_{(1)} \cdot \underbrace{\prod_{i=1}^k \mathcal{D}(\Lambda^\dagger_i)}_{(2)}   \end{equation}

The number of samples $N$ scales with the square of the cost of the observable (1) and the cost of channel adjoints (2), and is independent of the initial state.

Loosely speaking, Heisenberg propagation is efficient for \emph{any} separable input state or stabilizer mixture and supports a wider range of observables than Schr\"odinger propagation. However, it cannot capitalize on particularly noisy input states for a runtime improvement.

A version of Heisenberg propagation appears in \cite{sampling}, where they restrict operations to Clifford unitaries. Our work generalizes the technique to arbitrary quantum channels.

\section{\label{sec:states}Efficient Circuit Components}

In this section we study which input states, channels and observables (collectively `circuit components') can be simulated by Schr\"odinger, Heisenberg and stabilizer propagation without increasing runtime. This viewpoint helps address the practical question: ``Given a particular quantum circuit, which near-Clifford algorithm is best?'' 

Straightaway, if the quantum circuit is unitary then stabilizer rank techniques \cite{extent} are the best choice due to their superior accuracy and runtime. The primary advantage of propagation algorithms is their ability to support arbitrary circuit components with noise, measurement, and adaptivity. Despite their flexibility, the propagation algorithms vary significantly in their performance.

Since the number of samples scales as the product of the square of the cost of the components, a component occurring linearly many times with cost $>1$ demands exponential runtime. In the following, when we say an algorithm \textbf{supports} or \textbf{can handle} a component, we mean that the cost of the component is $\leq 1$, although the protocols can be applied to any component possibly inefficiently.

\subsection{Efficiency of Stabilizer Propagation}

For a self-contained description of stabilizer propagation see \cite{bennink, resource, seddon}. Just as the algorithms in section II decompose input states into a weighted sum of Pauli matrices, stabilizer propagation decomposes input states into a weighted sum of stabilizer states. A sampling process identical to equations \eqref{eq:sigma_hat} and \eqref{eq:c_hat} results in the number of samples required to be proportional to the square of the following normalization constant:

\begin{definition}
The \textbf{robustness of magic} $\mathcal{R}(\rho)$ of an $n$-qubit state $\rho$ is the outcome of a convex optimization program over real vectors $\vec q$:
\vspace{-1mm}
$$\mathcal{R}(\rho) = \min_{\vec q} \sum_i \lvert q_i \rvert \text{ s.t. } \rho = \sum_i q_i \ket{\phi_i} \bra{\phi_i} \text{ and } \sum_i q_i = 1,$$
\vspace{-4mm}

\noindent where $\{\ket{\phi_i}\}$ are the $n$-qubit stabilizer states. 

When $\mathcal{R}(\rho) = 1$ (the minimum value) then $\rho$ is a \textbf{stabilizer mixture}, since then the vector $\vec q$ is a probability distribution.
\end{definition}

Due to the sheer number of stabilizer states, evaluating $\mathcal{R}(\rho)$ for even small $n$ is very expensive. As stated in \cite{bennink}, evaluating the cost function for 3-qubit unitaries is impractical, although the performance can be improved for diagonal gates \cite{seddon}.

The performance of stabilizer propagation gives a lens for the non-Cliffordness of channels, studied extensively in \cite{seddon}. In the appendix, we expand on this work by modifying the protocol to support all \textbf{postselective channels} which include all trace preserving channels and all `reasonable' non-trace-preserving channels. There, we prove the following theorem:\\

\begin{theorem}\label{thm:stabprop}
Let $\Lambda$ be a postselective channel and let $\bar\phi_\Lambda$ be the channel's normalized Choi state. $\Lambda$ does not increase the number of samples required for stabilizer propagation if and only if $\mathcal{R}(\bar\phi_\Lambda) = 1$.
\end{theorem}

This establishes simple and flexible criteria for when a circuit component does not increase the runtime of stabilizer propagation: states $\rho$ are cheap when $\mathcal{R}(\rho) = 1$ and channels $\Lambda$ are cheap if $\mathcal{R}(\bar\phi_\Lambda) = 1$.

\subsection{Observables}

Observables encountered in practice are usually computational basis measurements, or operators with bounded norm that can be expressed as sums of not too many Pauli matrices.  Sometimes these observables are marginal: many of the qubits are not measured and traced out. Tracing out corresponds to measuring the identity observable, a kind of Pauli observable.

Stabilizer propagation outputs the inner product of the final observable with a stabilizer state. For all of the observables above, calculating inner products with stabilizer states is efficient: inner products with Pauli matrices can be obtained in $n^2$ time and marginal inner products with other stabilizer states in $n^3$ time \cite{tableau}. Crucially, these inner products remain bounded by the eigenvalues of the observable and thereby do not exponentially increase the range of the distribution.

Schr\"odinger propagation, which outputs the inner product with a Pauli matrix, does not have this property: although inner products between Pauli matrices are trivial to compute, the maximum inner product grows like $2^n$. Therefore, Schr\"odinger propagation is only viable when we are interested in the probability of measuring a particular state and only a constant number of discarded qubits. On the other hand, there exist contrived observables that only Schr\"odinger propagation can handle. If the observable is the tensor product of many non-stabilizer states, then neither Heisenberg propagation nor stabilizer propagation runs efficiently. (Indeed, calculating inner products of stabilizer states with tensor products of many non-stabilizer states is a key slow step in stabilizer rank techniques \cite{gosset, extent}.)

Heisenberg propagation applies the sampling method (1) (2) to the observable $E$, so cost is measured by $\mathcal{D}(E)$. The following facts, proven in \cite{resource}, show that Heisenberg propagation can handle the observables most common in quantum circuits.
\begin{proposition}
    $\mathcal{D}(\sigma) = 1$ for $\sigma \in \mathcal{P}_n$.
\end{proposition}
\begin{proposition}
    If $\ket{\phi}$ is a stabilizer state, then $\mathcal{D}(\ket{\phi}\bra{\phi}) = 1$.
\end{proposition}
\begin{proposition}
    $\mathcal{D}$ is multiplicative: $\mathcal{D}(A \otimes B) = \mathcal{D}(A) \cdot \mathcal{D}(B)$.
\end{proposition}

\subsection{Hyper-Octahedral States}

A central observation of this work is that Pauli matrix decompositions can produce similar simulational power as decompositions over stabilizer states. Here we show that despite their simplicity, Pauli matrix decompositions are \emph{more} powerful with regards to the input state of the circuit.  The number of samples required for Heisenberg propagation does not depend at all on the input state (19). For Schr\"odinger propagation we observe:
\begin{enumerate}
\renewcommand{\theenumi}{\Alph{enumi}}
    \item there exist states supported by Schr\"odinger propagation unsupported by stabilizer propagation, and
    \item sufficiently depolarized states can actively decrease the number of samples required.
\end{enumerate}

From the definition of the stabilizer norm, $\mathcal{D}$ can be viewed as the L1 norm of the Bloch vector $\vec x$ of $\rho$. The equation $||\vec x||_1 \leq 1$ defines the surface and interior of a hyper-octahedron, motivating the following definition.

\begin{definition}
\textbf{Hyper-octahedral states} $\rho$ satisfy $\mathcal{D}(\rho) \leq 1$. These states do not increase the number of samples for Schr\"odinger propagation.
\end{definition}

To see (B), we simply observe that the interior of the octahedron satisfies $\mathcal{D}(\rho) = ||\vec x||_1 < 1$. $\mathcal{D}$ is minimized at the $n$-qubit maximally mixed state where $\mathcal{D}(I/2^n) = 1/2^n$. The following result, proved in \cite{resource}, shows that all stabilizer mixtures are hyper-octahedral.

\begin{proposition}
For states $\rho$, $\mathcal{D}(\rho) \leq \mathcal{R}(\rho)$.
\end{proposition}

This fact classifies mixed states into three categories: stabilizer mixtures, non-stabilizer hyper-octahedral states, and magic states. For the single qubit, the first two categories coincide (the qubit stabilizer polytope is an octahedron). We plot a cross-section of the two-qubit Bloch sphere in FIG.~\ref{fig:ZIIZ}, showing that all of these categories are non-empty. FIG.~\ref{fig:states-pie} shows the relative quantity of these states according to the Hilbert-Schmidt measure. Stabilizer mixtures occupy a tiny fraction of all mixed states, whereas more than half are hyper-octahedral.

From the standpoint of quantum resource theories, hyper-octahedral states are interesting because they are similar to the `bound' states discussed in \cite{vcge12, hwve14, dh15, acb12}: they contain non-stabilizer mixed states that can be efficiently simulated. But unlike $\mathcal{R}$, tracing out qubits can increase $\mathcal{D}$. Hadamard eigenstates $\ket{H}$ are magic states that let Clifford circuits attain universal quantum computation, but $\ket{H} \otimes (I/2)$ is hyper-octahedral. Hyper-octahedral states are not bound for magic state distillation in the same sense as those in \cite{vcge12}: there are operations that can be simulated efficiently by stabilizer propagation that increase $\mathcal{D}$. Schr\"odinger propagation cannot simulate operations that increase $\mathcal{D}$.

\begin{figure}
    \centering

    \includegraphics[width=0.35\textwidth]{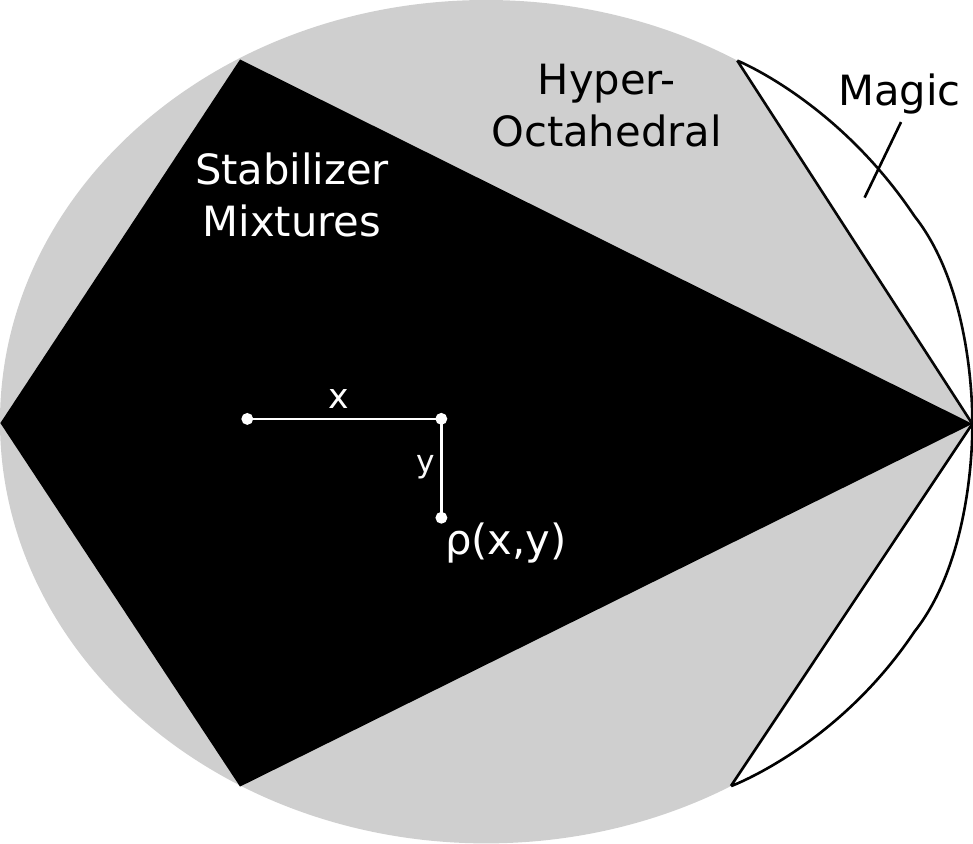}
    \caption{\label{fig:ZIIZ} Visualization of a cross section of the two-qubit Bloch sphere, given by: \vspace{1mm} \\ $\rho(x,y) = \frac{\sigma_{II}}{4} + x(\sigma_{XX} + \sigma_{ZZ} - \sigma_{YY}) + y(\sigma_{ZI} + \sigma_{IZ})$ }
\end{figure}

\begin{figure}
    \centering

    \includegraphics[width=0.35\textwidth]{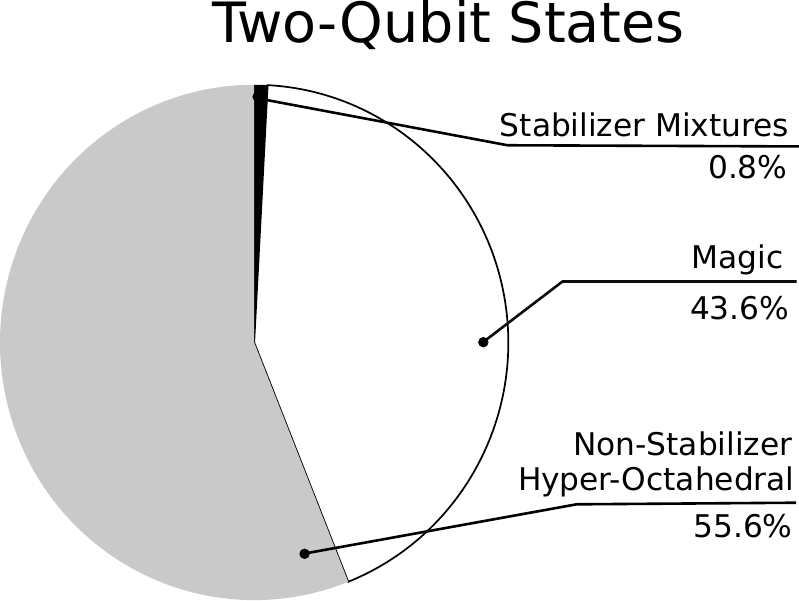}
    \caption{\label{fig:states-pie} Relative quantity of two-qubit mixed states, based on one million samples via the Hilbert-Schmidt measure. Hyper-octahederal states are plentiful for two-qubits, despite not existing for the single qubit.}
    \vspace{-3mm}
\end{figure}

\vspace{-4mm}
\subsection{Channel Classification}

While the classification of states gave rise to only three categories, the classification of channels is not so simple. FIG.~\ref{fig:Venn} shows eight categories, all of which are non-empty. Here are examples of each:
\begin{description}
    \setlength\itemsep{0mm}
    
    \item[M] Non-Clifford unitaries, such as the $T$ gate.
    \item[CSH] Clifford unitaries, measuring a qubit in a Pauli basis (without discarding it), and very depolarized non-Clifford unitaries.
    \item[SH] Mildly depolarized non-Clifford unitaries, e.g. the $T$ gate with fidelity $0.551 \lessapprox f \leq 2^{-1/2}$ (FIG.~\ref{fig:noisy_z_theta}).
    \item[C] Most adaptive Clifford gates: gates performed based on the outcome of a measurement (Proposition~\ref{thm:adapt}).
    \item[H] Any non-Pauli reset channel (Proposition~\ref{thm:reset}).
    \item[CH] Pauli reset channels \cite{bennink}.
    \item[S, CS] Channels adjoints for H, HC, respectively.
\end{description}

To obtain the relative proportions of these categories akin to FIG.~\ref{fig:states-pie} we leverage channel-state duality. Our definition of postselective channels in the appendix is specifically chosen to make the correspondence between two-qubit mixed states and qubit-to-qubit channels a bijection. We sample states according to the Hilbert-Schmidt measure and classify their corresponding channels. Most channels in practice are either unital, trace preserving or both. It is not obvious how to restrict sampling to these measure-zero subspaces. Instead, we sample from the full Hilbert Schmidt measure, and then project onto the Bloch-subspaces corresponding to unital and/or trace preserving channels. 

FIG.~\ref{fig:channel-pies} shows the resulting proportions. For qubit-to-qubit channels, Pauli propagation techniques permit simulation of a significant fraction of the circuit components which are a superset of those simulable by stabilizer propagation.  As before, it is not clear that this demonstrates that Pauli propagation is significantly more useful in practice, since most quantum circuits are dominated by a few specific types channels.

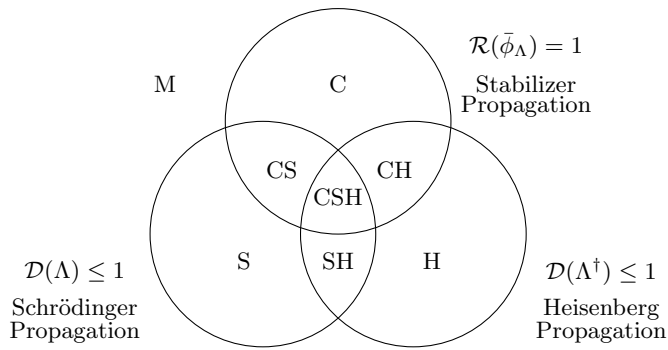
\begin{figure}[t]
    \begin{tikzpicture}
        \draw (3.5, 1.5) circle [radius = 1.5];;
        \draw (2.5, 0) circle [radius = 1.5];;
        \draw (4.5, 0) circle [radius = 1.5];;
        \node at (6, 2.5) {$\mathcal{R}(\bar \phi_\Lambda) = 1$};
        \node at (6, 2) {Stabilizer};
        \node at (6, 1.7) {Propagation};
        \node at (0, -0.5) {$\mathcal{D}(\Lambda) \leq 1$};
        \node at (0, -1) {Schr\"odinger};
        \node at (0, -1.3) {Propagation};
        \node at (7, -0.5) {$\mathcal{D}(\Lambda^\dag) \leq 1$};
        \node at (7, -1) {Heisenberg};
        \node at (7, -1.3) {Propagation};
        \node at (2.25, -0.35) {S};
        \node at (2.75, 0.85) {CS};
        \node at (3.5, 2) {C};
        \node at (3.5, 0.5) {CSH};
        \node at (3.5, -0.35) {SH};
        \node at (4.25, 0.85) {CH};
        \node at (4.75, -0.35) {H};
        \node at (1.2, 2) {M};
    \end{tikzpicture}

\caption{{\label{fig:Venn}}A Venn Diagram of quantum channels that illustrates our naming convention. The channels not efficient under any strategy are category M.}
\end{figure}

In the following we give evidence for the above examples. To do so, we phrase $\mathcal{D}(\Lambda)$ in terms of the Pauli transfer matrix of $\Lambda$.

\begin{definition}
    The Pauli Transfer Matrix (PTM) of a quantum channel $\Lambda$ taking $n$ qubits to $m$ qubits has elements $(R_\Lambda)_{ij} = 2^{-m}\Tr(\sigma_i\Lambda(\sigma_j))$ such that $\Lambda(\rho) = 2^{-n}\sum_{i, j} (R_\Lambda)_{ij} \sigma_i \Tr(\rho \sigma_j)$. We take $\sigma_{1} = I$.
\end{definition}

Intuitively, the columns of $R_\Lambda$ are the Bloch vectors of $\Lambda(\sigma_i)$. The following observations are useful and trivial to prove.

\begin{proposition}
$D(\Lambda) = \lVert R_\Lambda \rVert_1$, where $\lVert \cdot \rVert_1$ is the induced L1-norm, i.e. the largest column L1-norm.
\end{proposition}

\begin{proposition}
$R^T_\Lambda = R_{\Lambda^\dagger}$
\end{proposition}

\begin{corollary}{\label{thm:D_dag} }
$D(\Lambda^\dagger) = \lVert R_\Lambda \rVert_\infty$, where $\lVert \cdot \rVert_\infty$ is the induced L$\infty$-norm, i.e. the largest row L1-norm.
\end{corollary}

The PTM of a Clifford gate is a signed permutation matrix and the PTMs of Pauli basis measurements are signed permutations of $\text{diag}(1,1,0,0)$. Their Choi states are also readily shown the be stabilizer mixtures, so these channels are CSH as claimed.\\ 

\vfill
\onecolumngrid

\hspace{1mm}

\begin{table}[b]
\setlength{\tabcolsep}{6pt}
    \textbf{Relative Quantity of Qubit-To-Qubit Channels}
    \vspace{5mm}\\
    \includegraphics[width=\textwidth]{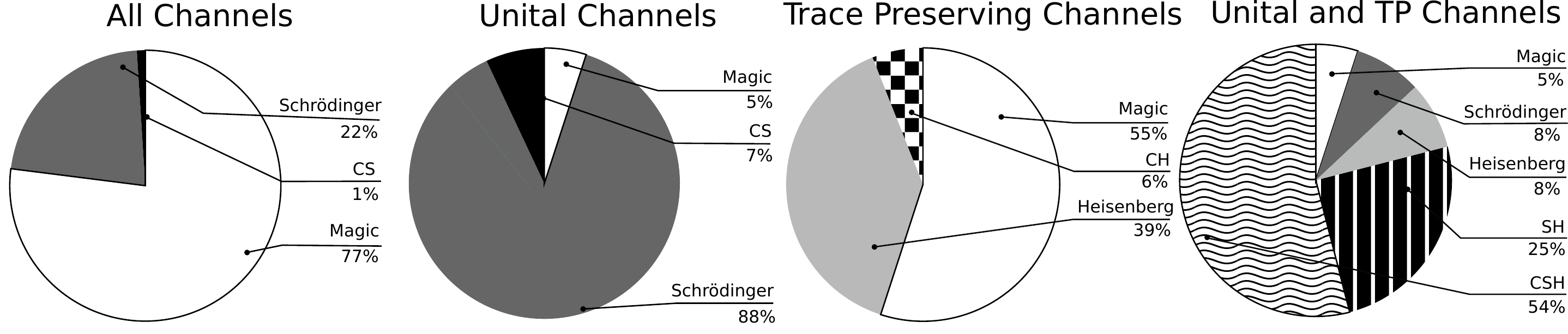}
    \refstepcounter{figure}
    \caption*{\label{fig:channel-pies}FIG. \arabic{figure}: Relative quantity of of qubit-to-qubit quantum channels, based on 100 000 random two-qubit density matrices obtained via the Hilbert-Schmidt measure. After obtaining the PTM we optionally set the first column or row to [1,0,0,0] to enforce unitality or trace preservation respectively \cite{gst}. We utilize the cvxpy library \cite{cvxpy} to compute $\mathcal{R}$ and use a tolerance of $10^{-6}$ throughout. }
\end{table}
\twocolumngrid

\hspace{1cm}

\clearpage

\subsection{Depolarized Rotations}

Many useful unitaries take the form $e^{-i\theta \sigma/2}$ with $\sigma \in \mathcal{P}_n$. Via some Clifford transformations these can be obtained from the qubit unitary $e^{-i\theta \sigma_Z/2}$. In this section we consider composing this unitary with depolarizing noise, obtaining a family of channels $\Lambda_{\theta,f}$ where $f$ is the fidelity.

The PTMs of the unitary $e^{-\theta\sigma_Z/2}$ and depolarizing noise are respectively:
$$R_{\theta} =
\begin{bmatrix}
1 & 0 & 0 & 0\\
0 & \cos \theta & -\sin \theta & 0\\
0 & \sin \theta & \cos \theta & 0\\
0 & 0 & 0 & 1
\end{bmatrix}\\
\hspace{5mm} R_f = 
\begin{bmatrix}
1 & 0 & 0 & 0\\
0 & f & 0 & 0\\
0 & 0 & f & 0\\
0 & 0 & 0 & f
\end{bmatrix}
$$

Composing these two channels simply involves multiplying the two PTMs, resulting in:
\begin{eqnarray}R_{\Lambda_{f, \theta}} &=&
\begin{bmatrix}
    1 & 0 & 0 & 0\\
    0 & f\cos \theta & -f\sin \theta & 0\\
    0 & f\sin \theta & f \cos \theta & 0\\
    0 & 0 & 0 & f
\end{bmatrix}\\
    \mathcal{D}(\Lambda_{f, \theta}) = \mathcal{D}(\Lambda_{r, \theta}^\dag) &=& \max\big(1, f\lvert \cos \theta \rvert + f\lvert \sin \theta \rvert\big)
\end{eqnarray}

\begin{figure}[b]
    \centering
    \includegraphics[width=0.4\textwidth]{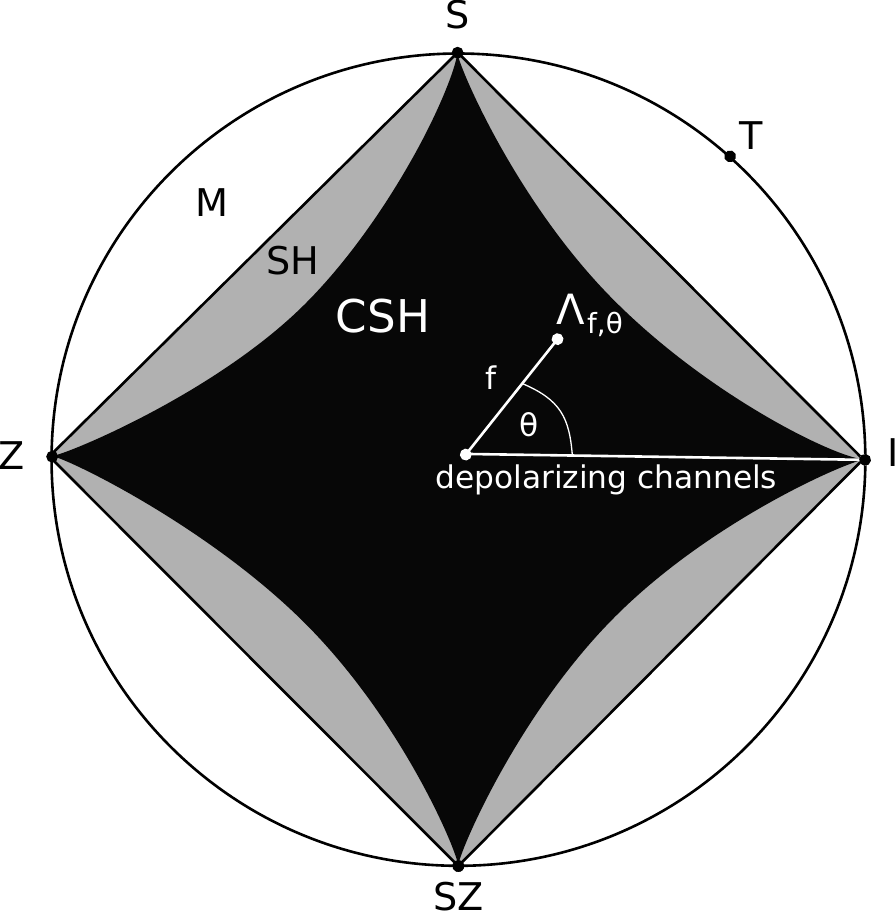}
    \caption{\label{fig:noisy_z_theta} Qubit quantum channels $\Lambda_{f,\theta}$ obtained by an application of the unitary  $e^{-i\theta \sigma_Z/2}$ followed by depolarizing noise with fidelity $f$. The region simulable by Pauli propagation (SH) is larger than that simulable by stabilizer propagation (CSH).}
    \vspace{-7mm}
\end{figure}

We plot the family in FIG.~\ref{fig:noisy_z_theta}, showing that there are channels simulable by Pauli propagation methods that are not simulable by stabilizer propagation. The boundary of $\mathcal{D} \leq 1$ given by $\lvert \cos \theta \rvert + \lvert \sin \theta \rvert = 1$ forms a diamond. The depolarized $T$ gate becomes SH when $f \leq 2^{-1/2} \approx 0.707$, and becomes CSH when $f \lessapprox 0.551$.
 
\subsection{Reset Channels}

Pauli reset channels can be described as projecting into the $+1$ eigenspace of some $\sigma \in \mathcal{P}_n $ as in \cite{bennink}. Alternatively we can use Clifford transformations to convert $\sigma$ to $\sigma_Z$, converting the channel to tracing out a single qubit and replacing it with $\ket{0}$. We generalize the notion of a reset channel $\Lambda_\rho$ to tracing out $n$ qubits and replacing them with an $n$-qubit state $\rho$. To make the channel trace preserving we write $\Lambda_\rho(\sigma) = \text{Tr}(\sigma) \cdot \rho$.

\begin{proposition}{\label{thm:reset}}
If $\Lambda_\rho$ is a reset channel, $\mathcal{D}(\Lambda^\dagger) = 1$.
\end{proposition}
\begin{proof}
The entries of the PTM of $\Lambda_\rho$ are the following:
$$
    (R_{\Lambda_\rho})_{ij} = 2^{-n}\Tr(\sigma_i\Lambda_\rho(\sigma_j)) =
\begin{cases}
    2^{-n}\Tr(\sigma_i \rho) & \sigma_j = I\\
0 & \sigma_j \neq I
\end{cases}
$$
All rows except for the first are zero. The entries are bounded $-1 \leq 2^{-n}\Tr(\sigma_i \rho) \leq 1$ and the top left entry is 1. Thus the maximum column L1 norm is 1, and Proposition \ref{thm:D_dag} tells us that $\mathcal{D}(\Lambda^\dag) = 1$.
\end{proof}

Observe that the first row is actually the Bloch vector of $\rho$ (including the identity component) scaled by $2^n$. So unless $\rho$ is the maximally mixed state the first row's L1 norm is $> 1$, so the channel is not simulable by Schr\"odinger propogation, and its adjoint is not simulable by Heisenberg propogation.

The Choi state of $\Lambda_\rho$ is $\frac{I}{2^n} \otimes \rho$, so $\Lambda_\rho$ is simulable by stabilizer propagation when $\rho$ is a stabilizer mixture.
\subsection{Adaptive Channels}

Adaptive channels consist of making a $\sigma_Z$ measurement, and then conditionally applying a channel based on the measurement outcome. While Pauli propagation techniques are stronger than stabilizer propagation in many respects, adaptive channels are their key weak point. This remains true even if the measured qubit is \emph{not} discarded, so we are not conflating the cost of tracing out qubits with the cost of adaptivity.

\begin{proposition}{\label{thm:adapt}}
    Let $\Lambda$ be a quantum channel with PTM $R_\Lambda$. Let $A(\Lambda)$ be the adaptive channel that conditionally applies $\Lambda$ based on a $\sigma_Z$ measurement on some qubit that is not discarded post-measurement. Then:
\begin{eqnarray}
    \mathcal{D}(A(\Lambda)) &=& 1 + \max_{i} \sum_{i\neq j} |R_{ij}| \leq 1 + \mathcal{D}(\Lambda^\dag)\\
    \mathcal{D}(A(\Lambda)^\dagger) &=& 1 + \max_{j} \sum_{i\neq j} |R_{ij}| \leq 1 + \mathcal{D}(\Lambda)
\end{eqnarray}

\end{proposition}

\begin{corollary} $A(\Lambda)$ is supported by Pauli propagation methods if and only if the PTM of $\Lambda$ is diagonal.
\end{corollary}

So Pauli propagation methods are not `closed under adaptivity': $A(\Lambda)$ can be non-simulable even if $\Lambda$ is simulable. Stabilizer propagation on the other hand \emph{is} closed under adaptivity.

\begin{proof}[Proof of Proposition~\ref{thm:adapt}]
    Let $\Lambda$ take $n$ qubits to $m$ qubits. The measurement of the first qubit projects into the space spanned by $I,\sigma_Z$ on the first qubit.

\begin{eqnarray*}
    A(\Lambda)(I \otimes \sigma_j) &=
    \begin{pmatrix}
        \sigma_j & 0\\
        0 & \Lambda(\sigma_j)
    \end{pmatrix} &=
    \begin{pmatrix}
        \sigma_j & 0\\
        0 & \sum_k R_{kj} \sigma_k
    \end{pmatrix}\\
    A(\Lambda)(\sigma_Z \otimes \sigma_j) &=
    \begin{pmatrix}
        \sigma_j & 0\\
        0 & -\Lambda(\sigma_j)
    \end{pmatrix} &=
    \begin{pmatrix}
        \sigma_j & 0\\
        0 & -\sum_k R_{kj} \sigma_k
    \end{pmatrix}
\end{eqnarray*}
The output remains in the space spanned by $I,\sigma_Z$ on the first qubit, so the only nonzero entries of the PTM are:
\begin{eqnarray*}
    \frac{1}{2^{m+1}}\text{Tr}\left( (I \otimes \sigma_i) \cdot A(\Lambda)(I \otimes \sigma_j)  \right) = \frac{1}{2}(\delta_{ij} + R_{ij}) \\
    \frac{1}{2^{m+1}}\text{Tr}\left( (\sigma_Z \otimes \sigma_i) \cdot A(\Lambda)(I \otimes \sigma_j)  \right) = \frac{1}{2}(\delta_{ij} - R_{ij}) \\
    \frac{1}{2^{m+1}}\text{Tr}\left( (I \otimes \sigma_i) \cdot A(\Lambda)(\sigma_Z \otimes \sigma_j)  \right) = \frac{1}{2}(\delta_{ij} - R_{ij}) \\
    \frac{1}{2^{m+1}}\text{Tr}\left( (\sigma_Z \otimes \sigma_i) \cdot A(\Lambda)(\sigma_Z \otimes \sigma_j)  \right) = \frac{1}{2}(\delta_{ij} + R_{ij})  
\end{eqnarray*}
Applying the definition of channel stabilizer norm:
\begin{eqnarray*}
    \mathcal{D}(A(\Lambda)) = \frac{1}{2} \max_i \sum_j \left(|\delta_{ij} + R_{ij}| + |\delta_{ij} - R_{ij}|\right)\\ = 1 + \max_i \sum_{i\neq j} |R_{ij}| \hspace{4mm}\square
\end{eqnarray*}
\renewcommand{\qedsymbol}{}
\end{proof}

\section{Numerical Results}

Algorithms based on Monte Carlo averages have favorable memory requirements and admit massive parallelization. We demonstrate these practical advantages via the performance of a GPU implementation written in CUDA \cite{cuda}. 

Following previous tests of near-Clifford algorithms \cite{extent} we simulate the Quantum Approximate Optimization Algorithm (QAOA) on E3LIN2 \cite{QAOAE3LIN2}. We generate $m$ random independent linear equations acting on three qubits $a,b,c \in [n]$ of the form $x_a \oplus x_b \oplus x_c = d_j$ for $j \in [m]$. Each qubit appears in at most $m/10$ equations. Let $\sigma^{(j)}_Z = \sigma_{Z,a} \otimes \sigma_{Z,b} \otimes\sigma_{Z,c}$ be $\sigma_Z$ acting on the qubits corresponding to equation $j$. Our goal is to estimate the observable
$$C = \frac{1}{2} \sum_{j \in [m]} (-1)^{d_j} \sigma^{(j)}_{Z} $$
since $C+m/2$ is the number of satisfied equations. We estimate the expectation of this observable with the state
$$\ket{\gamma,\beta} = e^{-i \beta B} e^{-i\gamma C}\ket{+^{\otimes n}}$$
where $B = \sum_{i \in [n]} \sigma_{X,i}$ and $\beta = \pi/4$. 

Heisenberg propagation is most appropriate for this problem, with performance $\mathcal{D}(C) = m/2$ and $\mathcal{D}(e^{\pm i\gamma \sigma^{(j)}_Z}) =|\sin \gamma| + |\cos \gamma| $. Although the unitary $e^{\pm i\gamma \sigma^{(j)}_Z}$ appears $m$ times in the circuit, at most $3(m/10 -1) + 1$ can act non-trivially on any term in $C$. Thus the accuracy of the simulation is given by:
$$\varepsilon_\text{Heis} = \frac{m}{\sqrt{2N}} \cdot \sqrt{\ln \frac{2}{\delta}} \cdot (|\sin \gamma| + |\cos \gamma|)^{3(m/10 -1) + 1} $$

As pointed out by \cite{extent}, a protocol by van den Nest \cite{nest} gives an \emph{efficient} Monte Carlo protocol for estimating $\langle C \rangle$ with error $\varepsilon_\text{Nest} =\frac{m}{\sqrt{N}} \cdot \sqrt{\ln \frac{2}{\delta}} $. We utilize the van den Nest estimate $\langle C\rangle_\text{Nest}$ to verify the Heisenberg propagation estimate $\langle C\rangle_\text{Heis}$. 

Writing effective CUDA applications demands careful memory management. Implementing stabilizer propagation via the Aaronson-Gottesman tableau algorithm would be a serious computer engineering task. In contrast, the increased simplicity of Pauli propagation algorithms permits a very simple implementation. We furthermore utilize bitwise operations to express the logic in a compact and efficient manner. Despite the better scaling it was ultimately necessary to also implement the van den Nest protocol in CUDA due to the sheer performance improvement over a Python implementation. 

\begin{figure}[b]
    \centering
    \includegraphics[width=0.5\textwidth]{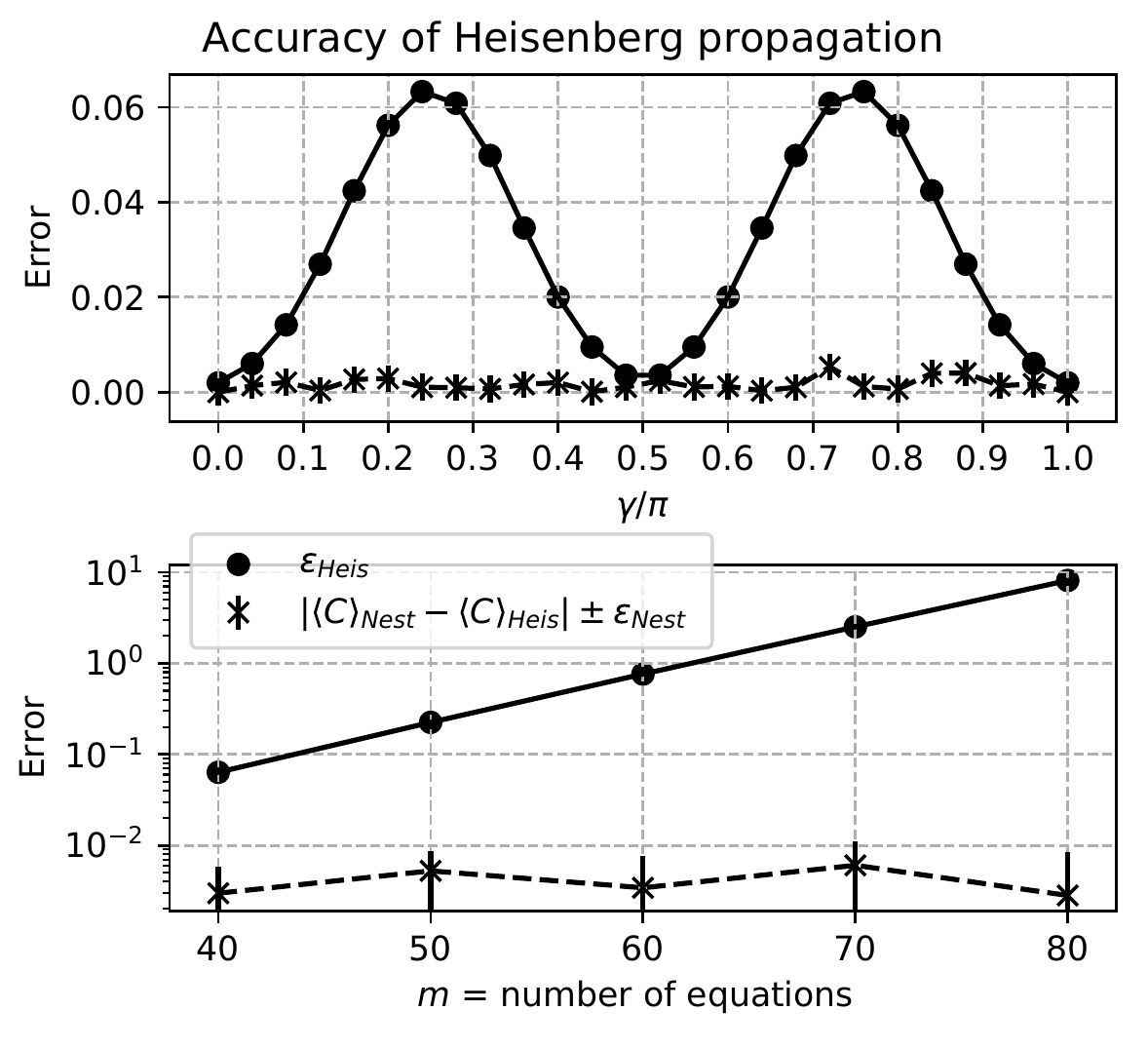}
    \caption{\label{fig:QAOA} Comparison of Hoeffding error bound $\varepsilon_\text{Heis}$ to error as estimated by the van den Nest protocol $|\langle C\rangle_\text{Nest} - \langle C\rangle_\text{Heis}|$ for 32 qubits. Top: $m = 40$ and varying $\gamma$. Bottom: $\gamma = \pi/8$ and varying $m$.}
    \vspace{-7mm}
\end{figure}

For every data point we collected $2^{30} \approx 1$ billion samples in 25 minutes using a laptop GPU (GeForce GTX 1050 Ti). We fix $n = 32$ qubits and $\delta = 0.01$ throughout, and vary $\gamma$ for a single instance with $m = 40$ equations (Figure~\ref{fig:QAOA}, top). Then we set $\gamma = \pi/8$, maximizing $\mathcal{D}(e^{\pm i\gamma \sigma^{(j)}_Z}) $ at $\sqrt{2}$, and perform a scaling analysis with instances up to $m = 80$ (Figure~\ref{fig:QAOA}, bottom). \\

Hoeffding's inequality gives a worst-case upper bound for the accuracy of the estimate, potentially very far from the actual error. This is the case here: for $m \gtrapprox 60$ we have $\varepsilon_\text{Heis} \geq 1$ predicting that  $\langle C\rangle_\text{Heis}$ is useless, but we observe that the actual error is $\leq 0.01$. Furthermore the actual accuracy does not seem to scale proportionally with $\varepsilon_\text{Heis}$ as we vary $\gamma$ and $m$.

\section*{Conclusion}
Recent interest in near-Clifford simulation \cite{bennink, pash, gosset, extent} and the (non-)contextuality of Clifford circuits \cite{hwve14, rbdobv15, rbdobv16, dovbr16} demonstrates that there is still much to be learned about embedding symmetry into Hilbert space. The qubit Clifford group appears different from the Clifford group in odd dimensions, where the discrete Wigner function \cite{gross} has led to well-behaved resource theories \cite{resource, vcge12, vmge13} and associated simulation algorithms \cite{sampling}. We observe that the qubit analogue of the Wigner function is just a Bloch vector, and our analysis of the resulting algorithms sheds further light into the differences between the even and odd-dimensional cases. Furthermore, the simplicity of Pauli propagation algorithms along with their improved performance for many quantum channels make them a compelling addition to near-Clifford simulation techniques.

\section*{Acknowledgements}
This work was supported by Dr. Scott Aaronson (UT Austin CS), who gave us invaluable advice throughout the project. We thank Dr. Antia Lamas-Linares (TACC) for giving us access to a plethora of supercomputing resources. We thank Dr. David Gross (Univ. of Colgone) as well as Dr. Earl Campbell, Dr. Mark Howard and James Seddon (Univ. of Sheffield) for useful suggestions and for coming up with creative names for some of the concepts introduced in this paper.  We thank David McKnight, DeVon Ingram and Adrian Trejo Nu\~nez for incisive editing feedback.

\appendix

\renewcommand\thetheorem{\Alph{section}.\arabic{theorem}}
\section{\label{sec:Postselective}Postselective Quantum Channels}

In this section we define postselective quantum channels. These include all trace preserving channels, and all `sensible' non-trace-preserving channels. Furthermore, there is a bijection between postselective channels taking $\mathcal{H}^A$ to $\mathcal{H}^B$ and density operators on $\mathcal{H}^B \otimes \mathcal{H}^A$, which is essential for FIG.~5 and Theorem~\ref{thm:stabprop}.

Completely positive maps $\Lambda$ with $0 \leq \text{Tr}(\Lambda(\rho)) \leq \text{Tr}(\rho)$ have an operational interpretation: the associated channels can `fail' or `abort' the computation by yielding 0. For example, let $\Lambda$ be the channel that measures in the $\sigma_Z$ basis and postselects on obtaining $\ket{0}$. Then $\Lambda(\ket{1}\bra{1}) = 0$, and $\Lambda(\ket{+}\bra{+}) = \frac{1}{2} \ket{0}\bra{0}$.

\begin{definition}
Let $\Lambda$ be a completely positive map from $\mathcal{H}^A$ to $\mathcal{H}^B$. Let $\ket{\text{Bell}_A} \in \mathcal{H}^A \otimes \mathcal{H}^A$ be a Bell state for $\mathcal{H}^A$, i.e. if $\{\ket{i}\}$ are a basis for $\mathcal{H}^A$ then:

\begin{equation}\ket{\text{Bell}_A} = \frac{1}{\sqrt{\text{dim}(\mathcal{H}^A)}} \sum_i \ket{i} \otimes \ket{i}  \end{equation}

    The \textbf{un-normalized Choi state} $\phi_\Lambda$ of $\Lambda$ is the resulting state when $\Lambda$ is applied to one half of $\ket{\text{Bell}_A}$.

\begin{equation}\phi_\Lambda = (\Lambda \otimes I)(\ket{\text{Bell}_A}\bra{\text{Bell}_A} ) \in \mathcal{H}^B \otimes \mathcal{H}^A\end{equation}

    $\text{Tr}(\phi_\Lambda$) of $\Lambda$ can be less than 1 if $\Lambda$ is not trace preserving. Let $\bar \phi_\Lambda = \phi_\Lambda/\text{Tr}(\phi_\Lambda)$ be the \textbf{normalized Choi state} with trace 1. This distinction is crucial.
\end{definition}

To calculate the output of a channel $\Lambda(\rho)$ given its Choi state $\phi_\Lambda$ we compute:
\begin{equation}\Lambda(\rho) = \text{dim}(\mathcal{H}^A)\cdot  \text{Tr}_A\left(\phi_\Lambda ( I \otimes \rho^T)\right)\label{eq:choiapply}\end{equation}

Crucially we use  $ \phi_\Lambda$, not $\bar \phi_\Lambda$. To explain why, consider a Choi state $\bar \phi_\Lambda = \ket{00}\bra{00}$. If we apply the equation above to $\bar \phi_\Lambda$ we obtain $\Lambda(\rho) = 2 \cdot \ket{0}\bra{0} \cdot \bra{0}\rho^T\ket{0}$, so $\Lambda(\ket{0}\bra{0}) = 2\ket{0}\bra{0}$ which makes no sense. The fact that $\phi_\Lambda$ is under-normalized takes care of this constant.

Given a normalized Choi state $\bar \phi_\Lambda$, e.g. $\ket{00}\bra{00}$, how do we determine $ \phi_\Lambda$? In general, $\phi_\Lambda$ is not unique. Consider channels $\Lambda(\rho)$ and $\Lambda'(\rho) = p\cdot0 + (1-p)\Lambda(\rho)$, i.e. $\Lambda'$ aborts with probability $p$ and otherwise applies $\Lambda$. Both channels have the same $\bar \phi_\Lambda$, but $\phi_{\Lambda'} = p \phi_\Lambda$.

However, $\Lambda'$ is somewhat silly: aborting the computation should be a tool for postselection and should not happen regardless of the input state. For all sensible channels there should exist an input state where the postselection succeeds with probability 1. To associate all $\bar \phi_\Lambda$ to a unique $\phi_\Lambda$ we restrict our attention to the following quantum channels.

\begin{definition}
    A completely positive map $\Lambda$ represents a \textbf{postselective quantum channel} if:
\begin{enumerate}
\item $\Lambda$ is trace-non-increasing: for all positive-semidefinite $\rho$, $\Lambda$ satisfies $0 \leq \text{Tr}(\Lambda(\rho)) \leq \text{Tr}(\rho)$,
\item the postselection can be satisfied: there exists a normalized pure state $\ket{\psi}$ such that $\text{Tr}(\Lambda(\ket{\psi}\bra{\psi})) = 1$.
\end{enumerate}
\end{definition}

Among these channels we can uniquely obtain $\phi_{\Lambda}$ from $\bar \phi_{\Lambda}$, so there is a bijection between normalized mixed states and postselective quantum channels. 
 Let  $ \phi_{\Lambda} = p_\Lambda \bar \phi_{\Lambda}$. Then:
\begin{equation} \frac{1}{p_\Lambda} = \text{dim}(\mathcal{H}^A)\cdot  \max_{\ket{\psi}} \text{Tr}\left( \bar\phi_{\Lambda} (I \otimes (\ket{\psi}\bra{\psi})^T)  \right) \label{eq:plambda}  \end{equation}

For example, if $\bar \phi_{\Lambda} = \ket{00}\bra{00}$ then $\ket{\psi} = \ket{0}$ maximizes $1/p_\Lambda$ at 2, so  $\phi_{\Lambda} = \frac{1}{2}\ket{00}\bra{00}$ and $\Lambda(\rho) = \ket{0}\bra{0} \cdot \bra{0}\rho^T\ket{0}$.
Incidentally, $p_\Lambda$ is the probability of postselection succeeding when $\Lambda$ is applied to the Bell state.

\section{\label{sec:R_channel_proof}Simulating Channels whose Choi States are Stabilizer Mixtures}

In this appendix we prove Theorem~\ref{thm:stabprop}: Stabilizer propagation can efficiently simulate a quantum channel $\Lambda$ if and only if the robustness of its Choi state $\mathcal{R}(\phi_\Lambda)$ is 1. This criterion also captures postselective quantum channels, and thereby all sensible non-trace-preserving channels.

All results of \cite{seddon} generalize neatly to postselective channels. Assuming familiarity with the work, the definition of magic capacity $\mathcal{C}(\Lambda)$ remains identical and the channel robustness $\mathcal{R}_*(\Lambda)$ can be obtained via convex optimization over linear combinations of \emph{un-normalized} Choi states of stabilizer channels. It is easy to see that Theorem 2, $\mathcal{R}(\phi_\Lambda) \leq \mathcal{C}(\Lambda) \leq \mathcal{R}_*(\Lambda)$, still holds. Their Lemma 2, $\mathcal{R}(\bar \phi_\Lambda) = 1$ implies $\mathcal{C}(\Lambda) = 1$, is our Theorem 3.2.

\newtheoremstyle{named}{}{}{\itshape}{}{\bfseries}{.}{.5em}{#1 \thmnote{#3}}
\theoremstyle{named}
\newtheorem*{namedthm}{Theorem}

\begin{namedthm}[\ref{thm:stabprop} (rephrased)]
Consider a postselective channel $\Lambda: \mathcal{H}^A \to \mathcal{H}^B$. The following statements are equivalent.
\begin{enumerate}
\item The channel's normalized Choi state $\bar \phi_\Lambda$ is a probabilistic mixture of stabilizer states, so $\mathcal{R}(\bar\phi_\Lambda) = 1$.

\item If $\Lambda$ is applied to any subset of the qubits of any large stabilizer state $\ket{\psi}$, one can efficiently sample from a probability distribution over stabilizer states and `abort' whose mean is the resulting state.
\end{enumerate}
\end{namedthm}

\begin{proof}[Proof: 2. implies 1.]
 Say a channel $\Lambda$ is simulable. Apply $\Lambda$ to one half of the state $\ket{\text{Bell}_A}$, a stabilizer state. The resulting Choi state is probabilistic mixture of stabilizer states and `abort':

\begin{eqnarray}\phi_\Lambda = p_0 \cdot 0 + \sum_i p_i \ket{\phi_i}\bra{\phi_i}\\
\bar\phi_\Lambda =  \frac{\phi_\Lambda}{\text{Tr}(\phi_\Lambda)}  = \frac{1}{1-p_0} \sum_i p_i \ket{\phi_i}\bra{\phi_i}\end{eqnarray}
Since $p_i / (1-p_0)$ is a probability distribution, $\bar\phi_\Lambda$ is also a probabilistic mixture of stabilizer states.
\end{proof}

\begin{proof}[Proof: 1. implies 2.] Say we are given
\begin{equation}\bar\phi_{\Lambda} = \sum_i p_i \bar\phi_{\Gamma_i}\end{equation}
where $\bar \phi_{\Gamma_i}$ are \emph{pure} stabilizer states with corresponding pure operations $\Gamma_i$. Our goal is to obtain an efficiently computable probability distribution over stabilizer states and `abort' of $\Lambda$ applied to some subset of the qubits of a stabilizer state $\ket{\psi}$. The channel acts on a constant number of qubits, so we can compute anything we want about it. The stabilizer state, however, may live in a Hilbert space of exponential dimension. Using $ \phi_{\Lambda} = p_\Lambda \bar \phi_{\Lambda}$:
\begin{equation} \phi_{\Lambda} = p_\Lambda \sum_i \frac{p_i}{p_{\Gamma_i}}  \phi_{\Gamma_i} \end{equation}
All of the quantities $p_\Lambda, p_i$ and $p_{\Gamma_i}$ can be obtained quickly. Now we apply (\ref{eq:choiapply}), but we extend $\Lambda$ and $\Gamma_i$ from the constant size Hilbert space to $\tilde\Lambda$ and $\tilde \Gamma_i$ which act on the large Hilbert space containing $\ket{\psi}$.
\begin{equation}\tilde\Lambda(\ket{\psi}\bra{\psi}) = p_\Lambda \sum_i \frac{p_i}{p_{\Gamma_i}} \tilde \Gamma_i(\ket{\psi}\bra{\psi})\end{equation}

    Crucially, $\tilde \Gamma_i(\ket{\psi}\bra{\psi})$, the right-hand side of (\ref{eq:choiapply}), is an inner product between pure stabilizer states $\phi_\Gamma$ and $\ket{\psi}$ and is therefore a pure stabilizer state that can be computed in polynomial time. Since $\tilde \Gamma_i$ may be non-trace-preserving, $\tilde \Gamma_i(\ket{\psi}\bra{\psi})$ may not be normalized. Let $\ket{\gamma_i}$ be the normalized pure stabilizer state:
    \begin{equation}
\ket{\gamma_i}\bra{\gamma_i} = \tilde \Gamma_i(\ket{\psi}\bra{\psi})\Big/ \text{Tr}(\tilde \Gamma_i(\ket{\psi}\bra{\psi}))\hspace{5mm}
    \end{equation}
We write $\tilde\Lambda(\ket{\psi}\bra{\psi})$ as a weighted sum over normalized pure stabilizer states $\ket{\gamma_i}$.
\begin{equation}\tilde\Lambda(\ket{\psi}\bra{\psi}) =  \sum_i p_\Lambda \frac{p_i}{p_{\Gamma_i}} \text{Tr}(\tilde \Gamma_i(\ket{\psi}\bra{\psi})) \cdot \ket{\gamma_i}\bra{\gamma_i}\hspace{5mm}\end{equation}

The weights are positive and one can see that they sum to less than 1 by taking the trace of both sides. Furthermore since $\bar\phi_{\Gamma_i}$ are pure stabilizer states, the number $ \text{Tr}(\tilde \Gamma_i(\ket{\psi}\bra{\psi})) $ and stabilizer state $\ket{\gamma_i}\bra{\gamma_i}$ are efficiently computable.

Thus, to simulate $\Lambda$ acting on $\ket{\psi}$ we sample:
\begin{eqnarray}\ket{\gamma_i}\bra{\gamma_i}  &\text{ w.p. }& p_\Lambda \frac{p_i}{p_{\Gamma_i}} \text{Tr}(\tilde \Gamma_i(\ket{\psi}\bra{\psi})) \\ 
 0 &\text{ w.p. }& 1- \sum_i p_\Lambda \frac{p_i}{p_{\Gamma_i}} \text{Tr}(\tilde \Gamma_i(\ket{\psi}\bra{\psi})).\hspace{1mm}\square\end{eqnarray}
 \renewcommand{\qedsymbol}{}
\end{proof}

\begin{thebibliography}{2}
    \bibitem{bennink} Ryan Bennink, Erik Ferragut, Travis Humble, Jason Laska, James Nutaro, Mark Pleszkoch, Raphael Pooser ``Unbiased Simulation of Near-Clifford Quantum Circuits'' \href{https://journals.aps.org/pra/abstract/10.1103/PhysRevA.95.062337}{Phys. Rev. A 95, 062337 } \href{https://arxiv.org/abs/1703.00111}{quant-ph/1703.00111} (2017)
    \bibitem{pash} Hakop Pashayan, Joel Wallmann, Steven Bartlett ``Estimating outcome probabilities of quantum circuits using quasiprobabilities'' \href{https://journals.aps.org/prl/abstract/10.1103/PhysRevLett.115.070501}{Phys. Rev. Lett. 115, 070501} \href{https://arxiv.org/abs/1503.07525}{quant-ph/1503.07525} (2015)
    \bibitem{sampling} Hakop Pashayan, Stephen D. Bartlett, David Gross ``From estimation of quantum probabilities to simulation of quantum circuits'' \href{https://arxiv.org/abs/1712.02806}{quant-ph/1712.02806} (2017)
    \bibitem{resource} Mark Howard, Earl Campbell ``Application of a resource theory for magic states to fault-tolerant quantum computing'' \href{https://journals.aps.org/prl/abstract/10.1103/PhysRevLett.118.090501}{Phys. Rev. Lett. 118, 090501} \href{https://arxiv.org/abs/1609.07488}{quant-ph/1609.07488} (2016)
    \bibitem{gosset} Sergey Bravyi, David Gosset ``Improved Classical Simulation of Quantum Circuits Dominated by Clifford Gates'' \href{https://journals.aps.org/prl/abstract/10.1103/PhysRevLett.116.250501}{Phys. Rev. Lett. 116, 250501} \href{https://arxiv.org/abs/1601.07601}{quant-ph/1601.07601} (2016)
    \bibitem{extent} Sergey Bravyi, Dan Browne, Padraic Calpin, Earl Campbell, David Gosset, Mark Howard ``Simulation of quantum circuits by low-rank stabilizer dedcompositions'' \href{https://arxiv.org/abs/1808.00128}{quant-ph/1808.00128} (2018)
    \bibitem{jozsa} Richard Jozsa, Maarten Van den Nest ``Classical simulation complexity of extended Clifford circuits'' \href{http://www.rintonpress.com/xxqic14/qic-14-78/0633-0648.pdf}{Quantum Information and Computation, Vol. 14, No. 7\&8 0633-0648} \href{https://arxiv.org/abs/1305.6190}{quant-ph/1305.6190} (2014)
    \bibitem{seddon} James Seddon, Earl Campbell ``Quantifying magic for multi-qubit operations'' \href{https://arxiv.org/abs/1901.03322}{quant-ph/1901.03322} (2019)
        
    \bibitem{gross} David Gross ``Hudson's Theorem for finite-dimensional quantum systems'' \href{http://aip.scitation.org/doi/10.1063/1.2393152}{Journal of Mathematical Physics 47, 122107} \href{https://arxiv.org/abs/quant-ph/0602001}{quant-ph/0602001} (2006)
    \bibitem{vcge12} Victor Veitch, Christopher Ferrie, David Gross, Joseph Emerson. ``Negative Quasi-Probability as a Resource for Quantum Computation'' \href{http://iopscience.iop.org/article/10.1088/1367-2630/14/11/113011/meta}{New J. Phys. 15 039502}, \href{https://arxiv.org/abs/1201.1256v4}{quant-ph/1201.1257} (2012)
    \bibitem{vmge13} Victor Veitch, Seyed Ali Hamed Mousavian, Daniel Gottesman, Joseph Emerson. ``The Resource Theory of Stabilizer Computation'' \href{http://iopscience.iop.org/article/10.1088/1367-2630/16/1/013009/meta;jsessionid=3E7FD419DE6442746A2EB15603239E00.c5.iopscience.cld.iop.org}{New J. Phys. 16, 013009}, \href{https://arxiv.org/abs/1307.7171}{quant-ph/1307.7171} (2013)
    \bibitem{hwve14} Mark Howard, Joel Wallman, Victor Veitch, Joseph Emerson. ``Contextuality supplies the magic for quantum computation'' \href{http://www.nature.com/nature/journal/v510/n7505/full/nature13460.html}{doi:10.1038/nature13460}, \href{https://arxiv.org/abs/1401.4174}{quant-ph/1401.4174} (2014)
    \bibitem{rbdobv15} Robert Raussendorf, Dan Browne, Nicolas Delfosse, Cihan Okay, Juan Bermejo-Vega. ``Contextuality and Wigner function negativity in qubit quantum computation'' \href{http://journals.aps.org/prx/abstract/10.1103/PhysRevX.5.021003}{Phys Rev X 5, 021003}, \href{http://arxiv.org/abs/1511.08506}{quant-ph/1511.08506}, (2015)
    \bibitem{rbdobv16} Juan Bermejo-Vega, Nicolas Delfosse, Dan Browne, Cihan Okay, Robert Raussendorf. ``Contextuality as a resource for qubit quantum computation'' \href{https://arxiv.org/abs/1610.08529}{quant-ph/1610.08529} (2016)
    \bibitem{dovbr16} Nicolas Delfosse, Cihan Okay, Juan Bermejo-Vega, Dan Browne, Robert Raussendorf. ``Equivalence between contextuality and negativity of the Wigner function for qudits'' \href{https://arxiv.org/abs/1610.07093}{quant-ph/1610.07093} (2016)
    \bibitem{love} Lucas Kocia, Peter Love ``Discrete Wigner Formalism for Qubits and Non-Contextuality of Clifford Gates on Qubit Stabilizer States'' \href{https://journals.aps.org/pra/abstract/10.1103/PhysRevA.96.062134}{Phys. Rev. A 96, 062134} \href{https://arxiv.org/abs/1705.08869}{quant-ph/1705.08869} (2017)

    \bibitem{bk04} Sergei Bravyi, Alexei Kitaev. ``Universal Quantum Computation with ideal Clifford gates and noisy ancillas'' \href{http://journals.aps.org/pra/abstract/10.1103/PhysRevA.71.022316}{Phys. Rev. A 71, 022316}, \href{https://arxiv.org/abs/quant-ph/0403025}{quant-ph/0403025} (2004)
    \bibitem{dh15} Hillary Dawkins, Mark Howard. ``Qutrit Magic State Distillation Tight in Some Directions'' \href{http://journals.aps.org/prl/abstract/10.1103/PhysRevLett.115.030501}{Phys. Rev. Lett. 115, 030501}, \href{http://arxiv.org/abs/1504.05965}{quant-ph/1504.05965} (2015)
    \bibitem{acb12} Hussain Anwar, Earl Campbell, Dan Browne. ``Qutrit Magic State Distillation'' \href{http://iopscience.iop.org/article/10.1088/1367-2630/14/6/063006/}{New J. Phys. 14, 063006},  \href{https://arxiv.org/abs/1202.2326}{quant-ph/1202.2326} (2012)
    \bibitem{enums} Markus Heinrich, David Gross. ``Robustness of Magic and Symmetries of the Stabiliser Polytope''   \href{https://arxiv.org/abs/1807.10296}{quant-ph/1807.10296} (2018)


    \bibitem{nest} M. Van den Nest ``Simulating quantum computers with probabilistic methods'' Quant. Inf. Comp. 11, 9-10 pp. 784-812 (2011) \href{https://arxiv.org/abs/0911.1624}{quant-ph/0911.1624} (2010)
    \bibitem{zhuperm} Huangjun Zhu. ``Permutation Symmetry Determines the Discrete Wigner Function'' \href{https://journals.aps.org/prl/pdf/10.1103/PhysRevLett.116.040501}{Phys. Rev. Lett. 116, 040501} \href{https://arxiv.org/abs/1504.03773A}{quant-ph/1504.03773} (2015)
    \bibitem{zhusic} Huangjun Zhu. ``Quasiprobability representations of quantum mechanics with minimal negativity'' \href{https://journals.aps.org/prl/abstract/10.1103/PhysRevLett.117.120404}{Phys. Rev. Lett. 117, 120404} \href{https://arxiv.org/abs/1604.06974}{quant-ph/1604.06974} (2016)
    \bibitem{preskill} John Preskill ``Fault-tolerant quantum computation'' \href{https://arxiv.org/abs/quant-ph/9712048}{quant-ph/9712048} (1997) 
    \bibitem{tableau} Scott Aaronson, Daniel Gottesman ``Improved Simulation of Stabilizer Circuits'' \href{https://journals.aps.org/pra/abstract/10.1103/PhysRevA.70.052328}{Phys. Rev. A 70, 052328 }  \href{https://arxiv.org/abs/quant-ph/0406196}{quant-ph/0406196}  (2004)
    \bibitem{nisq} John Preskill ``Quantum computing in the NISQ era and beyond'' \href{https://arxiv.org/abs/1801.00862}{quant-ph/1801.00862} (2018) 
    \bibitem{hoeffding} Wassily Hoeffding ``Probability Inequalities for Sums of Bounded Random Variables'' \href{https://www.jstor.org/stable/2282952?seq=1#page_scan_tab_contents}{Journal of the American Statistical Association, Vol. 58, No. 301, pp. 13-30} (1963)
    \bibitem{volume} Karol Zyczkowski, Hans-Juergen Sommers ``Hilbert–Schmidt volume of the set of mixed quantum states'' \href{http://iopscience.iop.org/article/10.1088/0305-4470/36/39/310/meta}{J. Phys. A 36, 10115-10130} \href{https://arxiv.org/abs/quant-ph/0302197}{quant-ph/0302197}
    \bibitem{gst} Daniel Greenbaum. ``Introduction to Quantum Gate Set Tomography'' \href{https://arxiv.org/abs/1509.02921}{quant-ph/1509.02921} (2015)
    \bibitem{suprem} Martinis et al. ``Characterizing Quantum Supremacy in Near-Term Devices'' \textit{Nature Physics 14, 595-600}, \href{https://arxiv.org/abs/1608.00263}{quant-ph/1608.00263} (2016)
    \bibitem{QAOAE3LIN2} Edward Farhi, Jeffrey Goldstone. ``A Quantum Approximate Optimization Algorithm Applied to a Bounded Occurrence Constraint Problem'' \href{https://arxiv.org/abs/1412.6062}{quant-ph/1412.6062v2} (2015)
    \bibitem{cvxpy} Steven Diamond, Stephen Boyd ``CVXPY: A Python-Embedded Modeling Language for Convex Optimization'' Journal of Machine Learning Research (2016)
    \bibitem{cuda} Nicholas Wilt ``The CUDA Handbook: A comprehensive Guide to GPU programming'' Addison-Wesley Professional; 1st edition, June 22 2013

\end{thebibliography}
\end{document}